\documentclass[journal,doublecolumn]{IEEEtran}
\usepackage{graphicx}
\usepackage{subfigure}
\usepackage{amsfonts,amsmath,latexsym,amssymb}
\usepackage{theorem}
\usepackage{cite}
\usepackage{color}
\newtheorem{Theorem}{Theorem}
\newtheorem{lemma}{\textbf{Lemma}}

\newtheorem{DD}{Definition}
\usepackage[lined,boxed,commentsnumbered, ruled]{algorithm2e}
\usepackage{fancyhdr}
\usepackage{lipsum} 
\usepackage{setspace}

\begin{document}

\title{Coding based Data Broadcasting for Time Critical Applications with Rate Adaptation}
\author{Xiumin Wang,~
        Chau Yuen~and~Yinlong Xu
\thanks{Xiumin Wang is currently with the School of Computer and Information, Hefei University of Technology, Hefei, China.
        E-mail: wxiumin@hfut.edu.cn.}
\thanks{Chau Yuen is with Singapore University of Technology and Design, Singapore. Email: yuenchau@sutd.edu.sg.}
\thanks{Yinlong Xu is with the School of Computer Science, University of Science and Technology of China, Hefei, China. Email: ylxu@ustc.edu.cn.}
\thanks{This research is partly supported by the International Design Center
(grant no. IDG31100102 and IDD11100101). It is also supported in
part by the National Natural Science Foundation of China (Grant No.
61300212 and 61073038), and P.D. Programs Foundation of Ministry of
Education of China (Grant No. 20130111120010).}
 } \maketitle

\begin{abstract}
In this paper, we dynamically select the transmission rate and design wireless network coding to improve the quality of services such as delay for time critical applications.
In a network coded system, with low transmission rate and hence longer transmission range, more packets may be encoded, which increases the coding opportunity. However, low transmission rate may incur extra transmission delay, which is intolerable for time critical applications.
We design a novel joint rate selection and wireless network coding (RSNC) scheme with delay constraint, so as to maximize the total benefit (where we can define the benefit based on the priority or importance of a packet for example)  of the packets that are successfully received at the destinations without missing their deadlines.
We prove that the proposed problem is NP-hard, and propose a novel graph model to mathematically formulate the problem.
For the general case, we propose a transmission metric and design an efficient algorithm to determine the transmission rate and coding strategy for each transmission.
For a special case when all delay constraints are the same, we study the pairwise coding and present a polynomial time pairwise coding algorithm that achieves
an approximation ratio of $1-\frac{1}{e}$ to the optimal pairwise coding solution, where $e$ is the base of the natural logarithm. Finally, simulation results demonstrate the superiority of the proposed RSNC scheme.

\end{abstract}

\section{Introduction}
With the increase in both wireless channel bandwidth and
computational capability of wireless devices, wireless networks can
be used to support time critical applications such as video
streaming or interactive gaming. Such time critical applications
require the data content to reach the destination node(s) in a
timely fashion, i.e., a delay deadline is imposed on packet
reception, beyond which the reception becomes useless (or invalid)
\cite{XTL2006Time-critical14,Acharya1998,Dykeman1986,Aksoy1999,Xuan1997}.
These constraints can be imposed either by applications or the
users. For example, many financial users are interested in the
up-to-minute stock quotes so as to react to dynamic and rapid
market. As another example, in wireless location-based services, the
queried information (e.g., the traffic jam) is valid within a local
area, as when the mobile user (e.g., vehicle user) leaves the area,
the information becomes useless \cite{Linglong}.

Recently, network coding becomes a promising approach to improve
wireless network performance \cite{Li2003,Sagduyu2007,Wang1,
Wang2,ACL+2000Network1216,KRH+2008XORs510,Fragouli2007,NTN+2009Wireless925,RCS2007minimum125,Wang3}.
Specifically, the work in \cite{KRH+2008XORs510} proposed the first
network coding based packet forwarding architecture, named {\em
COPE}, to improve the throughput of wireless networks. With COPE,
each node opportunistically overhears some of the packets
transmitted by its neighbors, which are not intended to itself. The
relay node can then intelligently XOR multiple packets and forward
it to multiple next hops with only one transmission, which results
in a significant throughput improvement. Another important work on
wireless network coding is index coding
\cite{Bar-Yossef,Bar-Yossef2,Alon2008}. In index coding, a
source/server node needs to send some packets over a wireless
broadcast channel to some destinations/clients, and initially each
destination holds a subset of packets (i.e., side information).
Recent works show that with network coding, the number of
transmissions required can be reduced significantly, which thus
improves the throughput.

In most recent works, network nodes always transmit packets at a fixed rate. However, most wireless systems are now capable of performing adaptive modulation to vary the link transmission rate in response to the signal to interference plus noise at the receivers. Transmission rate diversity exhibits a rate-range tradeoff: the higher the transmission rate, the shorter the transmission range for a given transmission power \cite{KV2009Is646}. To aid overhearing, one may use the lowest transmission rate, so as to successfully deliver packet to more receivers/overhearing nodes.
Although this may increase the coding opportunity, it may not yield good performance, especially for time critical applications, as the arrival times of the packets may be delayed due to lower transmission rate.

In the literature, a few works studied the relationships between adapting the transmission rate and the network coding gain \cite{KV2009Is646,CJH2010Joint2444,Kim2010,Ni2008}. The work in \cite{KV2009Is646} showed that compared with pure network coding scheme, joint rate adaptation and network coding is more effective in throughput performance. They also proposed a joint rate selection and coding scheme to minimize the sum of the uplink and the downlink costs in star network topology. The work in \cite{CJH2010Joint2444} mathematically formulated the optimal packet coding and rate selection problem as an integer programming problem, and proposed an efficient heuristic algorithm to jointly find a good combination of coding solution and transmission rate. There are only a few works considered the delay guarantee of packet receptions, which is especially important for time critical applications.

So far, the works in
\cite{Traskov,Yang2012,Yeow,ZX2010Broadcast6,Zhan2011,CDE1,CDE2}
considered the delay constraint of packet reception with network
coding. Specifically, \cite{Traskov} designed a jointly scheduling
of packet transmission and network coding to meet the restriction of
packet receptions in a multi-hop wireless network. Compared with
multi-hop transmission, \cite{Yang2012} studied a wireless
broadcasting scheduling over a one-hop communication. To meet the
hard deadline of the packets, they designed adaptive network coding
and formulated the problem as a Markov decision process. However,
they assume that all the packets have the same deadline. The work in
\cite{Yeow} aimed to optimize the delay of multicasting a data
stream from a sender to multiple one-hop receivers with network
coding. Based on queuing theory, they analyzed the delay performance
from both system and receiver perspectives. The works in
\cite{ZX2010Broadcast6,Zhan2011} also considered the delay
constraint of packet receptions over one-hop communication, and
proposed a coding scheme to minimize the number of packets that miss
their deadlines. All the above literatures worked well in their
designed settings. However, they all assume that the transmission
rates on all the links are the same and fixed.

\begin{figure}[t]
\begin{center}
\includegraphics[height=29mm,width=69mm]{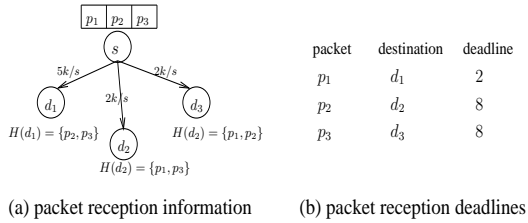}
\caption{Motivation illustration}\label{Fig.rsnc}
\end{center}
\end{figure}

Take Fig.~\ref{Fig.rsnc} as an example, where source node $s$ needs to transmit packet $p_1,p_2,p_3$ to node $d_1,d_2,d_3$
respectively. Fig.~\ref{Fig.rsnc}(a) gives the set of overheard packets $H(d_i)$ at destination $d_i$ (i.e. packets overheard by $d_i$ in previous transmissions, and available at $d_i$). Suppose that the size of each packet is $B=10k$, and the maximum transmission rates from $s$ to $d_1,d_2,d_3$ are $5k/s,2k/s$ and $2k/s$, respectively.
Fig.~\ref{Fig.rsnc}(b) shows the reception deadline of each required packet at its destination. For the current transmission, according to the work in \cite{KRH+2008XORs510,ZX2010Broadcast6}, $s$ will send the encoded packet $p_1\oplus p_2\oplus p_3$, as the most number of destinations can decode it. However, there is a problem for selecting the transmission rate at $s$. If $5k/s$ is selected, $d_2,d_3$ cannot successfully receive the packet, as the maximum transmission rates from $s$ to them are both $2k/s$. If $2k/s$ is selected, although all of the three receivers $d_1,d_2,d_3$ can receive and decode one ``wanted" packet, $p_1$  will miss its deadline at $d_1$, as its arrival time is $\frac{10k}{2k/s}=5s$.
As an alternative, we may choose to first send packet $p_1$ with transmission rate $5k/s$, where destinations $d_1$ will obtain a ``wanted" packet in $2s$. After this transmission, the encoded packet $p_2\oplus p_3$ can be sent with transmission rate $2k/s$, where destination $d_2$ and $d_3$ will receive and decode their ``wanted" packets after $7s$. Obviously, the latter solution is better than the first one, as no packet will miss the deadline.

In this paper, we study a new variant of index coding. By considering the impact of both transmission rate and network coding on the packet reception delay, we design a joint rate selection and network coding (RSNC) scheme for wireless time critical applications, so as to maximize the total benefits obtained by successfully receiving the packets without missing their deadlines. Here, benefit can be defined as the QoS or priority, and in this paper, we mainly set the benefit based on the priority or importance of a packet.
The main contributions of our paper can be concluded as follows:
\begin{itemize}
\item We propose a graph model, which considers both the heterogenous transmission rates and the deadlines of the packet receptions. Based on the graph model, we mathematically formulate the problem of maximizing the total benefits received by the packets that are successfully received at their destinations without missing deadlines, as an integer programming problem.
\item For each packet transmission, we propose a metric based on net benefit to determine the coded packet and the transmission rate. By considering the impact of the transmission rate on both delay and network coding gain, we also design an efficient algorithm to optimize the proposed metric.
\item We also consider a special case when the delay constraints for all the packets are the same. We study the pairwise coding solution for this special case, and present a polynomial time algorithm which achieves at least $1-\frac{1}{e}$ of the optimal pairwise coding solution.
\item We compare the performance of the proposed RSNC scheme with some existing algorithms. Simulation results show that the proposed scheme can significantly improve the total benefit obtained by the packets that are received without missing their deadlines.
\end{itemize}

The rest of the paper is organized as follows. In Section~\ref{Sec.formulation}, we define and formulate our problem.
The algorithm design for general joint rate selection and network coding scheme is given in Section~\ref{Sec.algorithm}. We study the pairwise coding solution for a special case in Section~\ref{Sec.special}. In Section~\ref{Sec.simulation}, we present the simulation results. Finally, we conclude the paper in Section~\ref{Sec.conclusion}.

\section{Problem Formulation}\label{Sec.formulation}
In this section, we first give the problem description and its complexity. Then, we introduce an auxiliary graph model, which can be used to design the algorithm. Finally, based on the graph model, we mathematically formulate the proposed problem. To ease understanding, the main notations are listed in Table~\ref{Table.notation}.

\begin{table}
\caption{\small{Main notations and their descriptions}}\center\label{Table.notation}
\begin{tabular}{|l|l|}
\hline
$B$ & The size of the packet\\
\hline
$D$ & The set of destination nodes\\
\hline
$d_i$ & The $i$-th destination in $D$\\
\hline
$H(d_i)$ & The set of available packets at destination $d_i$\\
\hline
$P$ & The set of $n$ packets\\
\hline
$p_j$ & The $j$-th packet in $P$\\
\hline
$R(d_i)$ & The set of required packets of $d_i$\\
\hline
$r(s,d_i)$ & The maximum transmission rate \\
& on transmission link from $s$ to $d_i$\\
\hline
$T_{i,j}$ & The deadline of packet $p_j$ required at $d_i$\\
\hline
$\alpha_{i,j}$ & The benefit of the packet $p_j$ at $d_i$\\
\hline
\end{tabular}
\end{table}

\subsection{Problem Description}
In this paper, we consider the application of network coding in
wireless broadcasting/multicasting. Without loss of generality, let
$s$ be the source/server node to send a data file to its destination
nodes in $D=\{d_1,d_2,\cdots,d_m\}$. Assume that the data file is
divided into $n$ packets in $P=\{p_1,p_2,\cdots,p_n\}$, and each
packet has the same size $B$. Suppose that initially, each
destination node has already stored a subset of packets in its
buffer (e.g., side information or received from previous
broadcasting) \cite{Zhan2011}. Let $H(d_i)$ be the set of available
packets at $d_i$, and $R(d_i)$ be the set of required packets by
$d_i$, i.e., $R(d_i)\subseteq P,H(d_i)\subseteq P$. For each packet
$p_j\in R(d_i)$, let $T_{i,j}$ be the reception deadline of packet
$p_j$ at node $d_i$. Let $\alpha_{i,j}$ be the benefit (i.e., profit
or importance) of packet $p_j$ at $d_i$. For example, the ``benefit"
$\alpha_{i,j}$ can be the profit of the packet $p_j$ obtained by
source node $s$ if $p_j$ is timely received at destination $d_i$ by
deadline $T_{i,j}$. Suppose that $r(s,d_i)$ is the maximum
transmission rate on link $(s,d_i)$, and only if the transmission
rate from $s$ to $d_i$ is no more than $r(s,d_i)$, the packet sent
from $s$ can be successfully received by $d_i$ \cite{KV2009Is646}.

Assume that $s$ knows the side information that each destination has
such that it can perform network coding operation. Such information
can be achieved by using {\em reception reports}, as introduced in
\cite{KRH+2008XORs510}. We also assume that node $s$ knows the
deadlines of the packet receptions at its receivers. As in COPE
\cite{KRH+2008XORs510}, only XORs coding is performed at the node in
our work. In addition, our broadcast channel is that at each time,
the source node broadcasts a single message to all the receivers at
one single rate, which is different from those in the
information-theoretic broadcast channel, where the source node can
send different packets to different receivers at different rates.

Our problem is that given the information of $H(d_i)$, $R(d_i)$, $T_{i,j}$, $\alpha_{i,j}$ and $r(s,d_i)$ for $\forall i,j$, we design the encoding strategy of the packets and select the transmission rate for each propagation, such that the total benefit of the packets that are successfully received at their destinations without missing their deadlines is maximum.

Let $z_{i,j}$ be $1$ if packet $p_j$ does not miss its deadline at $d_i$, otherwise, let it be $0$, where $p_j\in R(d_i)$. Thus, our objective is to maximize
\begin{eqnarray}
\sum_{d_i\in D}\sum_{p_j\in R(d_i)}\alpha_{i,j}z_{i,j}
\end{eqnarray}
In this paper, we refer such a problem of joint Rate Selection and Network Coding (RSNC) for time critical applications as RSNC problem. Note that if $\alpha_{i,j}=1$ for $\forall i,j$, the RSNC problem becomes to maximize the number of packets that are received without missing their deadlines, which is the case in our previous conference paper \cite{Wang2012}.

\begin{lemma}
The RSNC problem is NP-hard.
\end{lemma}
\begin{proof}
We consider a special case of the RSNC problem: $\alpha_{i,j}=1$,
$T_{i,j}$ is the same for $\forall i,j$, and the maximum
transmission rates on all the links are the same. Under the above
assumptions, the problem is reduced to NP index coding and the NP
hard nature of the problem is well known
\cite{Bar-Yossef,Bar-Yossef2}. Thus, the RSNC problem is also
NP-hard.
\end{proof}

According to the above lemma, we know that the complexity of finding the optimal solution of RSNC problem is exponential.

\subsection{Graph Model}\label{graph.model}
Although the graph model in \cite{ZX2010Broadcast6} works well for the case where the transmission rates on all the links are the same and fixed, it cannot be used directly for our RSNC problem. Here, we construct a novel graph model $G(V,E)$, which considers both the transmission rates and the packet reception deadlines.

We define $r_{min}(s,d_i|p_j)=\frac{B}{T_{i,j}}$ as the minimum transmission rate that can be used to meet the deadline of $p_j\in R(d_i)$ at $d_i$.
We add a vertex $v_{i,j}$ in $V(G)$, only if the following two conditions can be met.

(1) $p_j\in R(d_i)$;

(2) $r_{min}(s,d_i|p_j)\leq r(s,d_i)$.

Note that, if $r_{min}(s,d_i|p_j)> r(s,d_i)$, packet $p_j$ will definitely miss its deadline at $d_i$.
Thus, conditions (1) and (2) ensure that we add a vertex $v_{i,j}$ in $V(G)$ only if the ``wanted" packet $p_j$ may not miss its deadline at $d_j$.
That is, $V(G)=\{v_{i,j}|p_j\in R(d_i), r_{min}(s,d_i|p_j)\leq r(s,d_i)\}$.

Then, for any two different vertices $v_{i,j},v_{i',j'}\in V(G)$, there is an edge $(v_{i,j},v_{i',j'})\in E(G)$ if all the following conditions can be satisfied:

(a) $i\neq i'$;

(b) $j=j'$ or $p_j\in H(d_{i'})$ and $p_{j'}\in H(d_i)$;

(c) $r_{min}(s,d_i|p_j)\leq r(s,d_{i'})$ and $r_{min}(s,d_{i'}|p_{j'})\leq r(s,d_i)$.

We also define the weight of vertex $v_{i,j}$ as the benefit of packet $p_j$ at $d_i$, $\alpha_{i,j}$.
For any clique $Q=\{v_{i_1,j_1},v_{i_2,j_2},\cdots\}$ in $G$, let $P'=\{p_j|v_{i,j}\in Q\},D'=\{d_i|v_{i,j}\in Q\}$. Similar to the work in \cite{ZX2010Broadcast6}, if node $d_i\in D'$ successfully receives the encoded packet $p_{j_1}\oplus p_{j_2}\oplus \cdots\oplus p_{|P'|}$, where $p_{j_1}, p_{j_2},\cdots,p_{|P'|} \in P'$, $d_i$ can decode a ``wanted" packet $p_j$, where $v_{i,j}\in Q$.

Next, we will use an example to show the novelty of our graph model as compared to others in the literature, e.g., \cite{ZX2010Broadcast6}.
Still take Fig.~\ref{Fig.rsnc} as an example. The graph constructed by \cite{ZX2010Broadcast6} is shown in Fig.~\ref{Fig.example} (a). According to \cite{ZX2010Broadcast6}, any clique in the graph represents a feasible encoded packet. Thus, $p_1\oplus p_2\oplus p_3$ can be sent and its intended next hops are $d_1,d_2,d_3$, because $\{v_{1,1}$,$v_{2,2}$,$v_{3,3}\}$ forms a clique. As described before, it is not a good choice, as we cannot find a transmission rate to meet the deadlines of all the packets.
However, with our graph model shown in Fig.~\ref{Fig.example}(b), $p_1,p_2, p_3$ will not be encoded because vertices
$v_{1,1}$,$v_{2,2}$,$v_{3,3}$ do not form a clique in the graph. In addition, for the current transmission, the encoded packet derived from any clique in the graph can be sent without missing the deadlines at its intended destinations. For example, if $p_2\oplus p_3$, which is derived from the clique $\{v_{2,2},v_{3,3}\}$, is sent with the minimum of the possible transmission rates among $r(s,d_2)$ and $r(s,d_3)$, $2k/s$, its intended next hops $d_2,d_3$ can successfully decode the packets $p_2,p_3$ respectively without missing their deadlines.

\begin{figure}[t]
\begin{center}
\includegraphics[height=26mm,width=90mm]{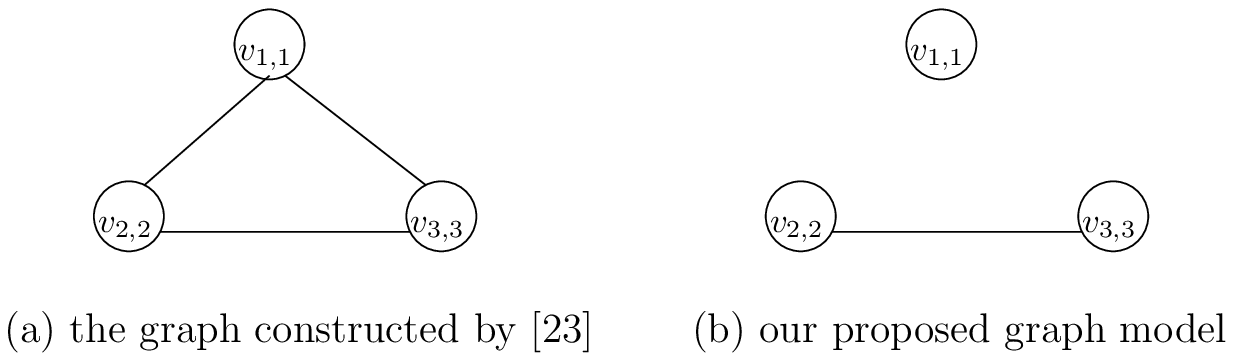}
\caption{Different graph model comparison}\label{Fig.example}
\end{center}
\end{figure}

For any clique $Q$ in the graph $G(V,E)$, we have the following lemma.
\begin{lemma}\label{lemma_graph}
In the current transmission, if the encoded packet $p_{j_1}\oplus p_{j_2}\oplus \cdots\oplus p_{|P'|}$, where $p_{j_1},\cdots,p_{|P'|} \in P'$ and $P'=\{p_j|v_{i,j}\in Q\}$, is sent with the transmission rate $r=\min\{r(s,d_i)|d_i\in D'\}$, it will be received by all the nodes in $D'$. Then, for each $v_{i,j}\in Q$, the packet $p_j$ will be decoded by $d_i$ without missing its deadline. In addition, the benefit obtained by such a transmission is the sum of the weights on the vertices in the clique, i.e., $\sum_{v_{i,j}\in Q}\alpha_{i,j}$.
\end{lemma}
\begin{proof}
Firstly, we can easily obtain that with transmission rate $r=\min\{r(s,d_i)|d_i\in D'\}$, all the receivers in $D'$ can successfully receive the sending packet. This is because the transmission rate $r$ must be lower than the maximum transmission rate from $s$ to any $d_i\in D'$.

Secondly, the constructed graph satisfies the decoding property given in \cite{ZX2010Broadcast6}. According to \cite{ZX2010Broadcast6}, if $p_{j_1}\oplus p_{j_1}\oplus \cdots\oplus p_{j_{|P'|}}$ is successfully received by $d_i \in D'$, $d_i$ can decode its ``wanted" packet $p_j$, where $v_{i,j}\in Q$.
Thus, any receiver $d_i\in D'$ can obtain a ``wanted" packet $p_j$ from $p_{j_1}\oplus p_{j_2}\oplus \cdots\oplus p_{j_{|P'|}}$ with transmission rate $r$, where $v_{i,j}\in Q$.

Thirdly, according to the condition (c), we have
\begin{eqnarray}
r_{min}(s,d_i|p_j)\leq \min_{d_i\in D'}\{r(s,d_i)\}=r
\end{eqnarray}
So, its arrival time at receiver $d_i$ is
\begin{eqnarray}
\frac{B}{r}\leq \frac{B}{r_{min}(s,d_i|p_j)}=\frac{B}{\frac{B}{T_{i,j}}}=T_{i,j}
\end{eqnarray}
In other words, the arrival time of the packet $p_j\in P'$ will not miss its deadline at its receiver $d_i\in D'$, where $v_{i,j}\in Q$.
\end{proof}

With Lemma~\ref{lemma_graph}, a clique $Q$ in the graph represents a feasible transmission solution for the current propagation, with the encoded packet $p_{j_1}\oplus p_{j_2}\oplus \cdots\oplus p_{j_{|P'|}}$, transmission rate $r=\min\{r(s,d_i)|d_i\in D'\}$, the propagation delay $\frac{B}{r}$, and the benefit $\sum_{v_{i,j}\in q}\alpha_{i,j}$.

\subsection{RSNC Formulation}\label{formulation}
While Lemma~\ref{lemma_graph} ensures that any encoding strategy based on any clique in the graph will be delivered within deadline for the current packet transmission, the transmission orders of the encoded packets, represented by the cliques in $G(V,E)$, is important for the timely packet receptions at their destinations.

For example, as shown in Fig.~\ref{Fig.example}(b) where we assume that $\alpha_{i,j}=1$ for $\forall i,j$, if we first schedule packet $p_1$ with transmission rate $5k/s$, represented by clique $\{v_{1,1}\}$, and then schedule packet $p_2\oplus p_3$ with transmission rate $2k/s$, represented by clique $\{v_{2,2},v_{3,3}\}$, all the packets will be received at their destinations without missing deadlines. The total benefit with such a solution is $\alpha_{1,1}+\alpha_{2,2}+\alpha_{3,3}=3$. However, if we first schedule packet $p_2\oplus p_3$, and then packet $p_1$, packet $p_1$ will miss its deadline at $d_1$. Correspondingly, the total benefit with the latter solution is only $\alpha_{2,2}+\alpha_{3,3}=2$.

Thus, our next task is to find a set of cliques in the graph and schedule the transmissions of the encoded packets represented by these cliques, so as to maximize the total benefit of the packets that are successfully received/decoded at their destinations without missing their deadlines.
Suppose that $Q_h=\{v_{i_1,j_1},v_{i_2,j_2},\cdots,\}$ is the $h$-th clique found in the graph, and the corresponding encoded packet represented by $Q_h$ is sent
as the $h$-th transmission at node $s$. We also assume that $P'_h=\{p_j|v_{i,j}\in Q_h\}$, $D'_h=\{d_i|v_{i,j}\in Q_h\}$. Thus, the packet sent by the $h$-th transmission at $s$ is $p_{j_1}\oplus p_{j_2}\oplus\cdots \oplus p_{j_{|P'_h|}}$ where $p_{j_1},p_{j_2},\cdots,p_{j_{|P'_h|}}\in P'_h$, and the transmission rate is $r_h=\min_{d_i\in D'_h}\{r(s,d_i)\}$. Let $T'_h$ be the transmission delay of the $h$-th transmission, i.e., $T'_h=\frac{B}{r_h}$.

We first define the following variant.
\begin{equation}
x_{i,j,h}=\left\{
\begin{aligned}
&1,\mbox{if vertex}~v_{i,j}~\mbox{is included in clique}~Q_h\\
&0,\mbox{otherwise}
\end{aligned}
\right.
\end{equation}
Then, we can formulate the RSNC problem based on the graph model as follows.
\begin{eqnarray}\label{formulation.obj}
\max_{\{Q_h\}} & &\sum_{d_i\in D}\sum_{p_j\in R(d_i)}\alpha_{i,j}z_{i,j}
\end{eqnarray}subject to \begin{eqnarray}
&&\sum_{h=1}^{|V(G)|}x_{i,j,h}=1,\forall v_{i,j}\in V(G)\label{ch.1}\\
&&x_{i,j,h}+x_{i',j',h}=1,\forall h\in \{1,2,\cdots,|V(G)|\},\\
&& \forall (v_{i,j},v_{i',j'})\notin E(G)\label{ch.2}\\
&&T'_h=\max_{v_{i,j}\in V}\{\frac{B*x_{i,j,h}}{r(s,d_i)}\},1\leq h\leq |V(G)|\label{ch.3}\\
&&\sum_{h=1}^{|V(G)|}(x_{i,j,h}*\sum_{j=1}^hT'_j)\leq T_{i,j}+\xi z_{i,j},\forall v_{i,j}\label{ch.4}\\
&&\sum_{h=1}^{|V(G)|}(x_{i,j,h}*\sum_{j=1}^hT'_j)\geq T_{i,j}-\xi (1-z_{i,j}),\forall v_{i,j}\label{ch.5}\\
&&z_{i,j}\in\{0,1\},\forall i,j\label{ch.6}
\end{eqnarray}where $\xi$ is a sufficient large constant.

In the above formulation, the term of the objective represents the total benefit of the packets that are successfully received at their destinations without missing their deadlines, which needs to be maximized.
Constraint~(\ref{ch.1}) denotes that each vertex in the graph can only belong to one clique. Constraint~(\ref{ch.2}) means that if there is no edge between vertex $v_{i,j}$ and $v_{i',j'}$, vertices $v_{i,j},v_{i',j'}$ cannot be in the same clique. Constraint~(\ref{ch.3}) gives the transmission delay for the $h$-th transmission, which is equal to the transmission delay with the minimum transmission rate among the rates from $s$ to all intended receivers. The sufficient large constant $\xi$ is used to guarantee that if $\sum_{h=1}^{|V|}(x_{i,j,h}*\sum_{j=1}^hT'_j)>T_{i,j}$, $z_{i,j}$ must be $1$, as denoted in Constraint~(\ref{ch.4}), and if $\sum_{h=1}^{|V|}(x_{i,j,h}*\sum_{j=1}^hT'_j)\leq T_{i,j}$, $z_{i,j}$ must be $0$, as denoted in Constraint~(\ref{ch.5}). Note that the arrival time of the packet in the $h$-th transmission should consist of both the waiting time of the previous $h-1$ transmissions and the transmission time of the $h$-th transmission, i.e., $\sum_{j=1}^h T'_j$.
Thus, Constraint~(\ref{ch.4}) and~(\ref{ch.5}) show that $z_{i,j}$ can be $0$ only if the arrival time of $p_j$ at $d_i$, i.e., $\sum_{h=1}^{|V|}(x_{i,j,h}*\sum_{j=1}^hT'_j)$, is no more than the reception deadline of packet $p_j$ at destination $d_i$.

With the above integer programming, we can get the optimal solution of RSNC problem. However, the computational complexity for the above integer programming is too high when the graph is large. Thus, we need to design an efficient heuristic algorithm to get a sub-optimal solution.

\section{Joint Rate Selection and Network Coding Algorithm}\label{Sec.algorithm}
Since every clique in the graph represents a feasible transmission strategy for the current transmission, we first design an algorithm to determine the encoding strategy and rate selection scheme for each packet propagation, by selecting a clique at a time. The whole transmission process will consist of multiple packets transmission/cliques selection, and will be introduced in Section~\ref{Sec.algorithm.design.whole}.

\subsection{Metric Consideration for Each Packet Propagation}\vspace{-0.03in}
First of all, in order to measure the ``goodness" of transmitting an encoded packet at a specific transmission rate for each packet propagation, it is necessary for us to adopt a reasonable metric which should take into account the impact of the transmission rate and the packet reception deadlines.
In this section, we shall design a metric, which not only satisfies as more requests as possible, but also minimizes the number of packets missing the deadlines after the current transmission.

For the current transmission, we define the following metric.
\begin{DD}
For an encoded packet $p_{j_1}\oplus p_{j_2}\oplus \cdots \oplus p_{j_L}$ sent with the transmission rate $r$, we define the metric $U$ as follows:
\begin{eqnarray}\label{equation.U}
U=\sum_{d_i\in D}\sum_{p_j\in R(d_i)} \alpha_{i,j} f_{i,j}- \sum_{d_i\in D}\sum_{p_j\in R(d_i)}\alpha_{i,j} l_{i,j}
\end{eqnarray}
\end{DD}where $f_{i,j}=1$ ($f_{i,j}=0$) denotes packet $p_j$ can (cannot) be decode/received by $d_i$ from the current transmission without missing its deadline, and $l_{i,j}=1$ ($l_{i,j}=0$) represents that packet $p_j$ will (will not) definitely miss its deadline after the current transmission.

The meaning of metric $U$ in Eq.~(\ref{equation.U}) can be explained
as follows. The first term $\sum_{d_i\in D}\sum_{p_j\in
R(d_i)}\alpha_{i,j} f_{i,j}$ denotes the benefit obtained from the
packets that are received without missing their deadlines from the
current transmission. The second term $\sum_{d_i\in D}\sum_{p_j\in
R(d_i)}\alpha_{i,j} l_{i,j}$ represents the lost for these packets
that will definitely miss their deadlines after the current
transmission. So, the metric $U$ denotes the net benefit obtained
from the current encoded packet and the transmission rate. For each
packet propagation, we aim to determine an encoded packet and select
the transmission rate $r$ that will maximize the metric $U$. In
other words, we would like to propose a greedy solution to optimize
the current performance only.

For the current transmission, given encoded packet $p_{j_1}\oplus p_{j_2}\oplus \cdots \oplus p_{j_L}$ and the transmission rate $r$, $f_{i,j}$ and $l_{i,j}$ can be both determined. For example,
$f_{i,j}$ is $1$ if and only if all the following conditions are met:
\begin{itemize}
\item $r\geq r(s,d_i)$, which means $d_i$ can successfully receive the sending packet;
\item $p_j\in R(d_i)$, which means $p_j$ is required by $d_i$;
\item All the other packets encoded in the current packet except $p_j$ are available at $d_i$, which is the decoding requirement of $p_j$ at $d_i$;
\item $\frac{B}{r}\leq T_{i,j}$, which shows the requirement of the reception deadline.
\end{itemize}
In addition, $l_{i,j}$ is $1$ if and only if for $\forall p_j\in R(d_i)$,
\begin{eqnarray} \label{eq.miss}
\frac{B}{r}+\frac{B}{r(s,d_i)}>T_{i,j}
\end{eqnarray}
Here, $\frac{B}{r}$ is the transmission delay of the current transmission, and $\frac{B}{r(s,d_i)}$ denotes the minimum delay to meet $p_j$'s deadline at $d_i$ in the next transmission. If the sum of the current transmission delay and the next minimum transmission delay is larger than the deadline of $p_j$ at $d_i$, $p_j$ will definitely miss its deadline, i.e., $l_{i,j}=1$.

Note that, the problem of maximizing the defined metric $U$ is also NP-hard. We can prove it by considering its special case: the transmission rates on all the links are the same,  the reception deadlines for all the packets are the transmission time of one packet, and each packet has the same benefit $\alpha_{i,j}$. The special case of maximizing the defined metric $U$ becomes to maximize the total number of the receivers that can decode one ``wanted" packet from the current encoded packet, which has been proved to be NP-hard in \cite{XiuminWang2010}.

\subsection{Heuristic Algorithm Design for Each Packet Propagation}\label{Sec.algorithm.design}
Although maximizing the defined metric $U$ is NP-hard, we can easily obtain the following observations, based on which we can design a heuristic algorithm.

1: Maximizing the first term of the metric $U$ is equal to find a maximum weight clique in the graph, where the weight at vertex $v_{i,j}$ is defined as the benefit $\alpha_{i,j}$.

2: The transmission rate is a parameter that adjusts the trade-off between delay and network coding gain. If $s$ uses a low transmission rate, more receivers can successfully receive the sending packet, and the current transmission may satisfy more receivers' requirements, denoted by the first term in $U$. However, low transmission rate means high transmission delay, which may cause more packets to miss their deadlines in the following transmissions, denoted by the second term in $U$.

Based on the above observations, we then design a heuristic algorithm for each packet propagation, by gradually increasing the transmission rate. Initially, the transmission rate is set to be no less than the lowest one from $s$ to its receivers. Let $TR=\{r(s,d_i)|d_i\in D\}$ be the set of available transmission rates from $s$ to all the destinations, and let $Tr_k$ be the $k$-th lowest rate in $TR$. As in Section~\ref{graph.model}, we construct the auxiliary graph with the given information.


In the $k$-th step, we restrict that the transmission rate used at $s$ must be no less than $Tr_k$. For $d_i$, if its maximum transmission rate from $s$ is less than $Tr_k$, it cannot successfully receive the sending packet. This restriction can be realized by omitting any vertex $v_{i,j}$ in $G(V,E)$ if $r(s,d_i)< Tr_k$. Then, we find the maximum weight clique in the subgraph $\{v_{i,j}|r(s,d_i)\geq Tr_k,v_{i,j}\in V(G)\}$, and adopt the transmission rate represented by the found clique. Each vertex $v_{i,j}$ in the found clique denotes that $p_j$ will be successfully received/decoded by $d_i$ without missing its deadline, for the given transmission rate. For each of the other packets that cannot be obtained at their receivers from the current transmission, we then judge whether it will definitely miss its deadline at its destinations, by (\ref{eq.miss}). Thus, in each step, we calculate $U$. Such process continues until all the rates in $TR$ are considered. Finally, we compare the values of $U$ obtained from each step and adopt the one with the largest value as the solution.
Note that, if there are more than one solution with the maximum value of $U$, we will choose the one with the smaller lost represented by the second term of (\ref{equation.U}).
The detailed of the algorithm is shown in Algorithm \ref{Alg.single}.

\subsection{Algorithm for the Whole Transmission Process}\label{Sec.algorithm.design.whole}
While Algorithm 1 in Section~\ref{Sec.algorithm.design} describes the encoding strategy of the packets and the selection of the transmission rate for every propagation, the whole transmission process will consist of multiple of such single process. We will first construct the graph $G(V,E)$ based on the model in Section~\ref{graph.model}, and the graph will be updated by removing the selected vertices in the found clique by Algorithm \ref{Alg.single}, and the vertex $v_{i,j}$ if $p_j$ will definitely miss its deadline at destination $d_i$. The packet reception deadlines for the packets also need to be updated after each transmission. The whole transmission process continues until the vertices set $V$ of $G$ becomes empty. The detail algorithm for the whole transmission is given in Algorithm \ref{Alg.whole}.

\begin{algorithm}[t]
\Begin
{

    $U_k=0$, $\forall k\in \{1,2,\cdots,|TR|\}$\;
    $f^k_{i,j}=l^k_{i,j}=0$, $\forall k\in \{1,2,\cdots,|TR|\},v_{i,j}\in V(G)$\;
    \For{$k\longleftarrow 1$ to $|TR|$}
    {
        find a max weight clique $Q_k$ in the subgraph $\{v_{i,j}|r(s,d_i)\geq Tr_k, v_{i,j}\in V\}$\;
        $f^k_{i,j}=1$, if $v_{i,j}\in Q_k$, for $\forall i,j$\;
        $r'_k=\min_{v_{i,j}\in Q_k}\{r(s,d_i)\}$\;
        \For{each $v_{i,j}\in V(G),v_{i,j}\notin Q_k$}
        {
             \If{ $\frac{B}{r'_k}+\frac{B}{r(s,d_i)}>T_{i,j}$}
             {
                $l^k_{i,j}=1$,
             }
        }
        $U_k=\sum_{d_i\in D}\sum_{p_j\in R(d_i)} \alpha_{i,j} f^k_{i,j}- \sum_{d_i\in D}\sum_{p_j\in R(d_i)}\alpha_{i,j} l^k_{i,j}$\;
    }
    add $Q_k$ into $\mathcal{Q}$ if $U_k$ is the maximum among $\{U_k|0\leq k\leq |TR|\}$\;
    $Q=\arg\min_{Q_k}\{\sum_{d_i\in D}\sum_{p_j\in R(d_i)}\alpha_{i,j} l^k_{i,j}|Q_k\in \mathcal{Q}\}$\;
    the current encoded packet is $\bigoplus_{v_{i,j}\in Q}p_j$\;
    the current transmission rate is $r=\min_{v_{i,j}\in Q}\{r(s,d_i)\}$;
    \For{each $v_{i,j}\in Q$}
    {
        delete $p_j$ from $R(d_i)$\;
        add $p_j$ to $H(d_i)$\;
    }
}
\caption{Algorithm design for one packet propagation process}\label{Alg.single}
\end{algorithm}

\begin{algorithm}[t]
\Begin
{
    construct graph $G(V,E)$\;
    \While{$V(G)$ is not empty}
    {
        conduct Algorithm 1 for the current packet propagation\;
        remove the selected clique from $G(V,E)$\;
        remove the vertex $v_{i,j}$ from $V(G)$ if $l_{i,j}=1$\;
        update the packet reception deadline, e.g., $T_{i,j}=T_{i,j}-\frac{B}{r}$\;
    }
}
\caption{Algorithm design for the whole packet transmission process}\label{Alg.whole}
\end{algorithm}

\subsection{Algorithm Complexity}
In this section, we will analyze the complexity of the proposed algorithm.

We first discuss the complexity of the algorithm for single packet propagation. According to Algorithm~\ref{Alg.single}, for each available transmission rate, the encoding strategy and rate selection is converted into finding a maximum weight clique in the defined subgraph. As finding a maximum weight clique in the graph is also an NP-hard problem, in this paper, we exploit a heuristic algorithm \cite{XiuminWang2010}, whose complexity is $O(|V(G)|^3)$. By trying all the available transmission rates in $TR$, the complexity of determining single packet propagation is $O(|V(G)|^3|TR|)$.

As the maximum number of transmissions is at least the number of packets in $P$, the complexity of the algorithm for the whole transmission process is $O(|V(G)|^3|TR||P|)$. In other words, the complexity of the above proposed algorithm is $O(|V(G)|^3|TR||P|)$.

\section{Pairwise Coding for the Same Delay Constraint}\label{Sec.special}

We now consider a special case when the delay constraints for all the packets are the same, e.g. in a video broadcasting, all packets have the same delay constraints to all the clients. As in \cite{KV2009Is646,Chaudhry2011}, we focus on a practically coding scheme, {\em pairwise coding} (or sparse coding in \cite{Chaudhry2011}), so as to decrease the encoding and decoding complexity. In this section, we first give a brief review on pairwise coding. Then, we design a greedy algorithm, and prove that the proposed algorithm achieves at least $(1-\frac{1}{e})$ of the optimal solution to pairwise coding, where $e$ is the base of the natural logarithm.

\subsection{Pairwise Coding}
With pairwise coding \cite{KV2009Is646,Chaudhry2011}, at most two requests of the destinations can be satisfied within one transmission, i.e., at most two packets are encoded together. As in Section~\ref{graph.model}, each clique in the graph model $G(V,E)$ represents a feasible packet transmission. In other words, with pairwise coding, the size of the clique selected for each transmission is at most two.

Without loss of generality, let $T$ be the delay constraint for all the packets.
Based on the graph $G(V,E)$ constructed in Section~\ref{graph.model}, we can enumerate all the cliques with at most size two as follows:
\begin{itemize}
\item For each vertex $v_{i,j}\in V(G)$, we add a unique clique $Q_h=\{v_{i,j}\}$;

\item For each edge $(v_{i,j},v_{i',j'})\in E(G)$, we add a unique clique $Q_{h'}=\{v_{i,j},v_{i',j'}\}$.
\end{itemize}
Let $\mathcal{Q}$ be the set of all the cliques defined above. For each element (vertex) included in the clique, we also define $\alpha_{i,j}$ as its weight. Then, the weight of the clique $Q_h\in \mathcal{Q}$, noted as $w_h$, can be defined as the sum of the weights on the elements covered by it, i.e., $w_h=\sum_{v_{i,j}\in Q_h}\alpha_{i,j}$. Note that in the following, the term ``weight" has the same meaning as the term ``benefit".
For each clique $Q_h\in \mathcal{Q}$, we also define the cost $c_h$ as the maximum transmission delay to satisfy the requests represented by the vertices in $Q_h$, i.e., $c_h=\max_{v_{i,j}\in Q_h}\{\frac{B}{r(s,d_i)}\}$.

With the above definition, the cost of clique $Q_h$ denotes the transmission delay of sending the encoded packet represented by $Q_h$, and the weight of clique $Q_h$ denotes the benefit of scheduling the encoded packet represented by $Q_h$ before deadline $T$.
Then, pairwise coding to maximize the total benefit is converted into finding a collection of cliques in $\mathcal{Q}^*\subseteq \mathcal{Q}$ such that total weight of all the selected cliques in $\mathcal{Q}^*$ is maximum, while the following conditions need to be satisfied:

(1) The total cost of the cliques in $\mathcal{Q}^*$ does not exceed a given budget $T$, i.e., $\sum_{Q_h\in \mathcal{Q}^*}c_h\leq T$, which denotes the delay constraint;

(2) For each two cliques $Q_{h}$ and $Q_{h'}$ in $\mathcal{Q}^*$, they must satisfy $Q_{h}\bigcap Q_{h'}=\emptyset$, which means that the satisfied request will be excluded in the following transmissions.

To ease understanding, we can formulate the above problem as follows.
\begin{eqnarray}\label{ILP}
\max_{\mathcal{Q}^*\subseteq \mathcal{Q}} \sum_{Q_h\in \mathcal{Q}^*} w_h
\end{eqnarray}subject to
\begin{eqnarray}\label{ILP.constraint}
&&\sum_{Q_h\in \mathcal{Q}^*}c_h\leq T\notag\\
&& Q_h\bigcap Q_{h'}=\emptyset,\forall Q_h,Q_{h'}\in \mathcal{Q}^*
\end{eqnarray}

We can easily prove that the above pairwise coding problem is NP hard, by reduction a classical NP problem, $0-1$ knapsack problem, to a special case of the proposed problem when there is no edge in graph $G(V,E)$.

\subsection{Algorithm Design}
We design a greedy algorithm to realize the pairwise coding, which is conducted with iterations.
To facilitate the further discussion, we define the following parameters.
\begin{itemize}
\item $\mathcal{Q}^*_k$: the collection of cliques selected in the first $k$ iterations, where $\mathcal{Q}^*_k\subseteq \mathcal{Q}$.
\item $C(\mathcal{Q}^*_k)$: the sum of the cost for the cliques in $\mathcal{Q}^*_k$.
\item $W(\mathcal{Q}^*_k)$: the total weight of the elements covered by the cliques in $\mathcal{Q}^*_k$.
\item $w^*_{k}(Q_h)$: the total weight of the elements covered by clique $Q_h$, but not covered by any clique in $\mathcal{Q}^*_{k}$.
\end{itemize}
Without loss of generality, we assume that the cost of each clique in $\mathcal{Q}$ is no more than $T$, since the cliques whose costs are more than $T$ must not belong to any feasible solution.

Initially, suppose that $U^*=\mathcal{Q}$ and $Q_0=\emptyset$. Our algorithm is conducted with iterations.
\begin{itemize}
\item In the $k$-th iteration, we consider the clique $Q_h$ in $U^*$ that maximizes the ratio $\frac{w^*_{k-1}(Q_h)}{c_{h}}$.
\item If adding cost $c_h$ to $C(\mathcal{Q}^*_{k-1})$ is more than $T$, discard the clique $Q_h$ from $U^*$.
\item If adding cost $c_h$ to $C(\mathcal{Q}^*_{k-1})$ is less than $T$, we then check if $Q_h$ has intersection with the selected cliques in $\mathcal{Q}^*_{k-1}$. If there is no intersection, add the clique $Q_h$ to the found collections, i.e., $\mathcal{Q}^*_{k}=\mathcal{Q}^*_{k-1}\bigcup \{Q_h\}$. Correspondingly, we delete $Q_h$ from $U^*$.
\item If there is intersection and the intersection set is not empty, there must exist another clique $Q_{h'}$ left in $U^*$, which includes the elements covered by clique $Q_h$ but not covered by any clique in $\mathcal{Q}^*_{k-1}$, i.e., $Q_{h'}=Q_h-Q_h\bigcap \{v_{i,j}|v_{i,j}\in Q_{h''}{\mbox{ and }}Q_{h''}\in \mathcal{Q}^*_{k-1}\}$. We then add clique $Q_{h'}$ to the selected collection, i.e., $\mathcal{Q}^*_k=\mathcal{Q}^*_{k-1}\bigcup Q_{h'}$, and delete both $Q_h$ and $Q_{h'}$ from $U^*$.
\item If there is intersection but the intersection set is empty, we just delete $Q_h$ from $U^*$.
\item The above process continues until $U^*=\emptyset$.
\end{itemize}
Finally, we compare the total weight of the elements covered by the selected cliques so far and the maximum weight of the clique in $\mathcal{Q}$, and keep the maximum one as the solution. The detail algorithm is shown in Algorithm~\ref{greedy2}.

\begin{algorithm}[t]
\Begin{

    $\mathcal{Q}^*_k\leftarrow \emptyset,C(\mathcal{Q}^*_k)=0,w^*_h(Q_h)=0, U^*\leftarrow \mathcal{Q}$, for $\forall k, h$\;
    $k=0$\;
    \While{$U^*\neq \emptyset$}
    {
        $k=k+1$\;
        Select $Q_h\in U^*$ that maximizes $\frac{w^*_{k-1}(Q_h)}{c_h}$\;
        \If{$C(\mathcal{Q}^*_{k-1})+c_h> T$}
        {
            $\mathcal{Q}^*_{k}\leftarrow \mathcal{Q}^*_{k-1}$\;
            $C(\mathcal{Q}^*_k)=C(\mathcal{Q}^*_{k-1})$\;
        }
        \Else
        {
            \If{there is no interaction between $Q_h$ and cliques in $\mathcal{Q}^*_k$}
            {
                $\mathcal{Q}^*_{k}\leftarrow \mathcal{Q}^*_{k-1}\bigcup \{Q_h\}$\;
                $C(\mathcal{Q}^*_k)=C(\mathcal{Q}^*_{k-1})+c_h$\;
            }
            \Else
            {
                \If{there exists another clique $Q_{h'}\in U^*$ that covers elements in $Q_h$ but not in any clique in $\mathcal{Q}^*_k$}
                {
                    $\mathcal{Q}^*_{k}=\mathcal{Q}^*_{k-1}\bigcup \{Q_{h'}\}$\;
                    $C(\mathcal{Q}^*_k)=C(\mathcal{Q}^*_{k-1})+c_{h'}$\;
                    $U^*\leftarrow U^*-Q_{h'}$\;
                }
            }
        }
        $U^*\leftarrow U-Q_h$\;
    }
    Consider the clique $Q_h$ that maximize $w_h$ over $\mathcal{Q}$\;
    If $W(\mathcal{Q}^*_k)\geq w_h$, output $\mathcal{Q}^*_k$, otherwise, output $\{Q_h\}$\;
} \caption{Algorithm design for pairwise coding}\label{greedy2}
\end{algorithm}

\subsection{Performance Analysis}
Before analyzing the performance of our algorithm, we first give a brief review on a classic budgeted set coverage problem \cite{Khuller1999}, which is similar to our problem and will be used in the following analysis.

\subsubsection{Budgeted Set Coverage Problem} The objective of the budgeted set coverage problem is to find a collection of the sets in $\mathcal{Q}$ such that the total weight of the elements covered by the selected sets is maximum while the total cost of all the selected sets is no more than a given value $T$. Compared with our problem, the budgeted set coverage problem does not require that the selected sets (i.e., cliques in our case) should be disjoint with each other. In the following presentation, we use the cliques and sets to denote the same meaning.

The algorithm in \cite{Khuller1999} is conducted with iterations. Let $\mathcal{Q}'_k\subseteq \mathcal{Q}$ be a collection of the cliques selected in the first $k$ iterations. Let $w'_{k}(Q_h)$ denote the total weight of the elements covered by $Q_h$, but not covered by any clique in $\mathcal{Q}'_{k}$. Initially, suppose that $U'=\mathcal{Q}$.
The main idea of the algorithm in \cite{Khuller1999} is that in each iteration $k$, the unconsidered clique $Q_h\in U'$ that maximizes the ratio of weight $w'_{k-1}(Q_h)$ to its cost (i.e., $\frac{w'_{k-1}(Q_h)}{c_h}$) is considered. If adding $c_h$ exceeds $T$, discard $Q_h$, otherwise, add $Q_h$ to the collection of selected cliques so far. Also, delete the considered clique $Q_h$ from $U'$. The above operation continues until all the cliques in $U'$ are considered.

Although the algorithm in \cite{Khuller1999} conducts well for the budgeted set coverage problem, it cannot be used to solve our problem.
This is because our problem requires the elements included in each pair of cliques in $\mathcal{Q}'$ to be disjoint with each other, as denoted in Constraint (2),
which differentiates from the budgeted set coverage problem.

\subsubsection{Analysis of the Proposed Pairwise Coding Solution}
Without loss of generality,
let $Q_{j'_t}$ and $Q_{j^*_t}$ be the $t$-th selected cliques with algorithm in \cite{Khuller1999} and our Algorithm~\ref{greedy2} respectively. Let $l'_t$ and $l^*_t$ be the index of the iteration in which clique $Q_{j'_t}$ and clique $Q_{j^*_t}$ are selected respectively. We assume that $\mathcal{Q}'=\{Q_{j'_1},Q_{j'_2},\cdots,Q_{j'_{q'}}\}$ is the collection of all the selected cliques with algorithm in \cite{Khuller1999}, and $\mathcal{Q}^*=\{Q_{j^*_1},Q_{j^*_2},\cdots,Q_{j^*_{q^*}}\}$ is the collection of all the selected cliques with Algorithm~\ref{greedy2}. Correspondingly, we let $S'_k$ and $S^*_k$ be the set of elements covered by the cliques in $\mathcal{Q}'_k$ and $\mathcal{Q}^*_k$ respectively.

Without loss of generality, we assume that each clique $Q_{j'_t}\in \mathcal{Q}'$ (or $Q_{j^*_t}\in \mathcal{Q}^*$) satisfies that $w^*_{l'_{t}-1}(Q_{j'_t})>0$ (or $w^*_{l^*_{t}-1}(Q_{j^*_t})>0$), otherwise, adding it has no contribution in increasing the total benefit (i.e., weight) and thus does not need to be considered.
Let $Q_{i'_{k}}$ and $Q_{i^*_{k}}$ be the original cliques considered in the $k$-th iteration whose ratio $\frac{w'_{k-1}(Q_{i'_{k}})}{c_{i'_{k}}}$, $\frac{w^*_{k-1}(Q_{i^*_{k}})}{c_{i^*_{k}}}$ are the maximum among all the unconsidered cliques so far in $U'$ and $U^*$ with algorithm in \cite{Khuller1999} and Algorithm~\ref{greedy2} respectively.

By comparing $Q_{i'_{k}}$ and $Q_{i^*_{k}}$, we can obtain the following lemma.
\begin{lemma}\label{lemma1}
At the $k$-th iteration, the clique $Q_{i'_{k}}$ considered by the algorithm in \cite{Khuller1999} is the same as the clique $Q_{i^*_{k}}$ considered by our Algorithm~\ref{greedy2}, i.e., $Q_{i'_{k}}=Q_{i^*_{k}}$, where $1\leq k\leq\min\{l'_{|\mathcal{Q}'|}, l^*_{|\mathcal{Q}^*|}\}$.
\end{lemma}
\begin{proof}
See Appendix~\ref{appendix_lemma1}.
\end{proof}

We can also obtain that
\begin{lemma}\label{lemma2}
At each $k$-th iteration, if the algorithm in \cite{Khuller1999} selects (does not select) a clique, the Algorithm~\ref{greedy2} must also select (not select) a clique, where $1\leq k\leq\min\{l'_{q'}, l^*_{q^*}\}$.
\end{lemma}
\begin{proof}
See Appendix~\ref{appendix_lemma2}.
\end{proof}

According to Lemma~\ref{lemma1} and Lemma~\ref{lemma2}, we can get the following Lemma.
\begin{lemma}\label{lemma2.a}
After the $k$-th iteration, the cliques in the collections $\mathcal{Q}'_k$ and $\mathcal{Q}^*_k$ cover the same set of the elements, i.e., $S'_k=S^*_k$, where $1\leq k\leq\min\{l'_{q'}, l^*_{q^*}\}$.
\end{lemma}
\begin{proof}
See Appendix~\ref{appendix_lemma2.a}.
\end{proof}

Based on the above lemmas, we can further obtain the following lemma.
\begin{lemma}\label{lemma3}
The number of cliques added in $\mathcal{Q}'$ must be no more than the number of cliques added in $\mathcal{Q}^*$, i.e., $q'\leq q^*$, and $l_{q'}\leq l_{q^*}$.
\end{lemma}
\begin{proof}
See Appendix~\ref{appendix_lemma3}.
\end{proof}

With the above lemmas, we can compare the weights of the elements covered by the two algorithms as follows.
\begin{lemma}\label{lemma4}
The total weight of the elements covered by $\mathcal{Q}'$ must be less than the total weight of the elements covered by $\mathcal{Q}^*$, i.e., $W(\mathcal{Q}')\leq W(\mathcal{Q}^*)$.
\end{lemma}
\begin{proof}
See Appendix~\ref{appendix_lemma4}.
\end{proof}

Based on the above results, we then prove the approximation ratio as follows.
\begin{Theorem}\label{theorem.ratio}
The total benefit achieved with our Algorithm \ref{greedy2} is at least $1-\frac{1}{e}$ of the optimal solution with pairwise coding.
\end{Theorem}
\begin{proof}
See Appendix~\ref{appendix_theorem.ratio}.
\end{proof}

\section{Simulation Results}\label{Sec.simulation}
In this section, we demonstrate the effectiveness of our RSNC scheme through simulations.
We randomly generate a set of available packets in $H(d_i)$ and the ``wanted" packets in $R(d_i)$ at destination $d_i\in D$, where $H(d_i)\bigcap R(d_i)=\emptyset$. The maximum transmission rate from $s$ to $d_i$ is randomly selected in $[rmin,rmax]$, and the packet reception deadline is randomly generated in $[Tmin,Tmax]$.

For comparison purpose, we mainly include two baseline algorithms,
namely, {\em DSF (deadline smallest first) coding} algorithm \cite{ZX2010Broadcast6} and {\em SIN-1} algorithm \cite{XTL2006Time-critical14}. DSF coding algorithm does not consider the heterogenous transmission rates on the links, and in each time slot, it always finds the maximum weight clique in the defined graph. SIN-1 algorithm always sends the packet with the minimum ``SIN-1" in each transmission, where ``SIN-1" of packet $p_j$ is defined as the ratio of the most urgent deadline of the requests for $p_j$ to the total number of requests for $p_j$. We also include Random Linear Network Coding and Index Coding as baseline algorithms in Section \ref{RLNC}.

In the simulation, we compare the total benefit obtained by timely receiving the required packets at their destinations, under different transmission schemes. For each setting, we present the average result of 200 samples.

\subsection{The Impact of the Transmission Rate}
We first investigate the impact of the transmission rate on the performance of the designed Algorithm~\ref{Alg.whole} for the whole transmission process. We randomly generate the benefit of each packet (i.e., $\alpha_{i,j}$) from $0.5$ to $2$.
We also set $n=m=10,Tmin=10,Tmax=50$ and vary the scale of the transmission rates, i.e., $[rmin, rmax]$.

As shown in Fig.~\ref{sim.rate}, with our RSNC scheme, the total benefit obtained by timely receiving the required packets at the destinations is much higher than the other two schemes. This is because RSNC not only considers the packet reception deadline for each packet, but also utilizes the heterogenous transmission rates from $s$ to the receiver nodes. We also can see that with the increase of the transmission rates, the total benefit increases. The reason is that higher transmission rates incur less transmission delay, which can satisfy more timely receptions in the following transmissions.

\begin{figure}[t]
\begin{center}
\includegraphics[height=46mm,width=65mm]{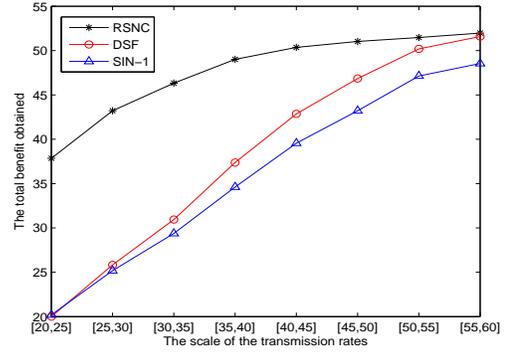}
\caption{The impact of the transmission rate on the performance of the whole transmission process.}\label{sim.rate}
\end{center}
\end{figure}

\subsection{The Impact of the Number of Destinations $m$}
\begin{figure}[t]
\begin{center}
\includegraphics[height=45mm,width=80mm]{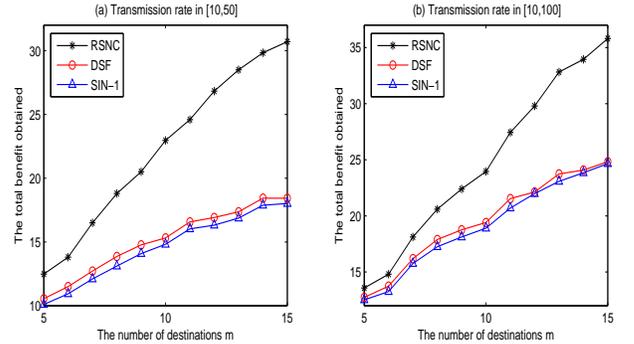}
\caption{The total benefit obtained by timely receiving the required packets vs. the number of destinations $m$.} \label{sim.m}
\end{center}
\end{figure}

We then investigate the impact of the number of destinations $m$ and the transmission rates on the deadline miss ratio. The benefit of each packet (i.e., $\alpha_{i,j}$) is randomly generated from $0.5$ to $2$. We set $n=10,Tmin=10,Tmax=50$ by varying $m$ in $[5,15]$ for $rmin=10,rmax=50$ and $rmin=50,rmax=100$.

As shown in Fig.~\ref{sim.m}, with our RSNC scheme, the total weight obtained by timely receiving the required packets at the destinations, is much more than with other schemes. We can also see that the DSF algorithm does not show significant gain over SIN-1 algorithm. This is because, although with network coding in DSF, more packets can be combined together, the encoded packet may still miss its deadline at some destinations, due to inappropriate transmission rate used.
In addition, with the increase of $m$, the total obtained benefit increases, as more requests need to be satisfied, which gives higher chance to obtain more benefit.
By comparing Fig.~\ref{sim.m}(a) and (b), we observe that, with the increase of the transmission rates, the total benefit increases, similar to Fig.~\ref{sim.rate}.

\subsection{The Impact of the Number of Packets $n$}
\begin{figure}[t]
\begin{center}
\includegraphics[height=45mm,width=90mm]{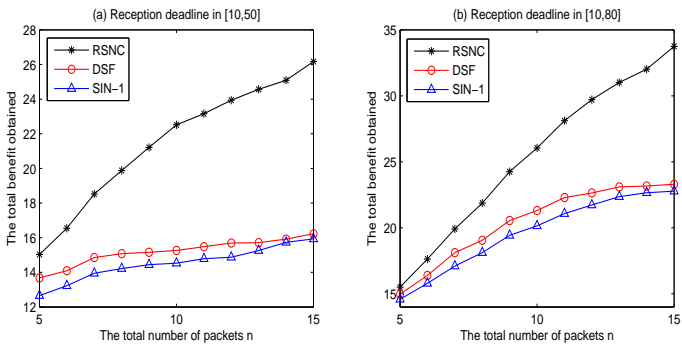}
\caption{The total benefit obtained by timely receiving the required packets vs. the total number of packets $n$.} \label{sim.n}
\end{center}
\end{figure}

We now investigate the impact of the total number of packets $n$ and the reception deadlines on the total benefit obtained by timely receiving the required packets at the receiver nodes. The benefit of each packet (i.e., $\alpha_{i,j}$) is randomly generated from $0.5$ to $2$. We set $m=10, rmin=10,rmax=50$ by varying $n$ in $[10,40]$ for the cases of $Tmin=10, Tmax=50$ and $Tmin=10,Tmax=80$.

From Fig.~\ref{sim.n}, we can see that our proposed RSNC scheme achieves the largest benefit while at the same time ensure timely receiving the required packets.
In addition, with the increase of $n$, the total benefit increases. This is because more requests need to be satisfied
at node $s$, which thus has higher chance to get more benefit. From Fig.~\ref{sim.n}, it is easy to see that the total benefit with $Tmax=80$ is smaller than with $Tmax=50$. It is reasonable because with the increase of the deadlines, less packet will lose its deadline at its destination node.

\subsection{The Deadline Miss Ratio}
\begin{figure}[t]
\begin{center}
\includegraphics[height=45mm,width=90mm]{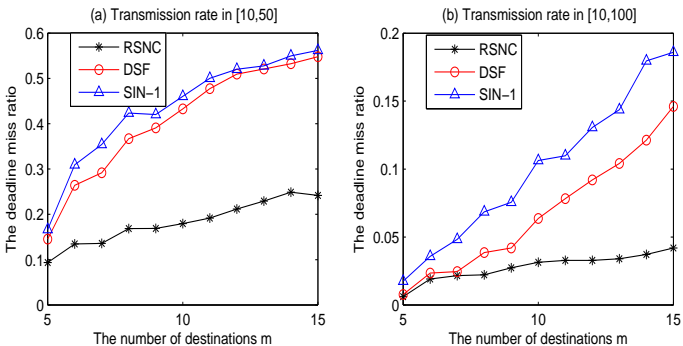}
\caption{The deadline miss ratio vs. the number of destinations.} \label{sim.deadline}
\end{center}
\end{figure}

We now investigate the performance of the deadline miss ratio, which is defined as the ratio of the number of the requests that miss their deadlines at the destinations, to the total number of the requests. In this case, we set every packet with the same benefit $\alpha_{i,j}=1$, so the total benefit directly relates to the total number of packets that are successfully delivered within the deadline. We also set $n=10,Tmin=10,Tmax=50$ by varying $m$ in $[5,15]$ for $rmin=10,rmax=50$ and $rmin=50,rmax=100$.

As shown in Fig.~\ref{sim.deadline}, the deadline miss ratio with our RSNC scheme is much lower than with other schemes. With the increase of $m$, the gain of our RSNC scheme increases. We can also see that the DSF algorithm does not show significant gain over SIN-1 algorithm. This is because, although with network coding in DSF, more packets can be combined together, the encoded packet may still miss its deadline at some destinations, due to inappropriate transmission rate used. From Fig.~\ref{sim.deadline}, we see that, with the increase of $m$, the deadline miss ratio increases. The reason is that there are more packets to be sent at $s$ within the same deadline scale.

\subsection{The Impact of $\alpha_{i,j}$}

We now study the impact of the benefit $\alpha_{i,j}$ on the
performance of the proposed algorithm. To fairly comparison, the
total number of requests are set to be $40$, and each request has
the same delay constraint. In addition, we set two kinds of
benefits: $\alpha_A$ and $\alpha_B$, where the first $20$ requests
are with benefit $\alpha_A$, while the other $20$ requests are with
benefit $\alpha_B$. We also set $\alpha_A=1$ and vary $\alpha_B$
between $[1,5]$.

To study the impact of the benefits on the performance of the
proposed algorithm, we define {\em the successful ratio for requests
with higher benefit} as the ratio of the number of requests that are
with benefit $\alpha_B$ and timely received at the destinations, to
the total number of requests that are timely received at the
destinations. As shown in Fig.~\ref{sim.alpha}, when
$\alpha_A=\alpha_B=1$, the successful ratio for requests with higher
benefit is almost $0.5$. This is because, with the same
benefit/priority, the number of successful requests with $\alpha_A$
is almost the same as the number of successful requests with
$\alpha_B$. We can also see that, with the increase of $\alpha_B$,
the successful ratio for requests with higher benefit increases
significantly, as the requests with higher benefit have higher
priority to be scheduled.

\begin{figure}[t]
\begin{center}
\includegraphics[height=48mm,width=65mm]{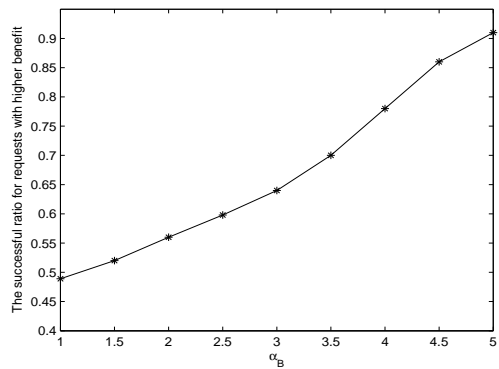}
\caption{The impact of the benefits of the packets.}\label{sim.alpha}
\end{center}
\end{figure}

\subsection{Comparison with Random Linear Network Coding and Index Coding}\label{RLNC}
To compare the performance of the proposed RSNC algorithm with random linear network coding and index coding, we investigate the performance of the deadline miss ratio. We set $n=10,Tmin=10,Tmax=50$ by varying $m$ in $[5,15]$ for $rmin=10,rmax=100$.

\begin{figure}[t]
\begin{center}
\includegraphics[height=48mm,width=65mm]{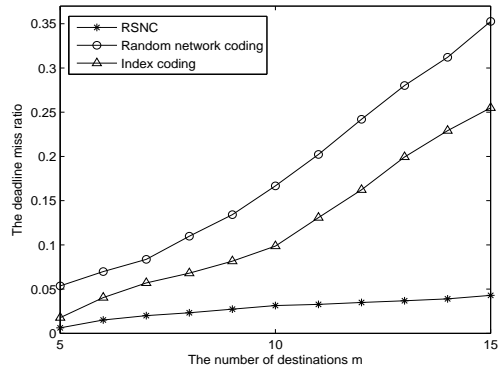}
\caption{Comparison with random linear network coding and index coding.}\label{sim.rlnc}
\end{center}
\end{figure}

As shown in Fig.~\ref{sim.rlnc}, the deadline miss ratio with the proposed RSNC scheme is much lower than that with random linear network coding and index coding. This is because either random linear network coding or index coding considers the deadline of packet reception. Particularly, with random linear network coding, the destination node cannot decode the native packet until receiving the full-rank encoded packets, which further delay the reception/decoding of the native packets. We can also see that with the increase of $m$, the gain of our RSNC scheme increases.

\subsection{Pairwise Coding}
In this subsection, we investigate the performance of pairwise coding for a special case when the packet reception deadline is the same for each packet. We compare the proposed Algorithm~\ref{greedy2} with the optimal pairwise coding obtained with ILP by Eq.~(\ref{ILP}) and Eq.~(\ref{ILP.constraint}).
In this simulation, we set $m=10,n=10,B=100k,rmin=10k/s,rmax=100k/s$. The benefit for timely receiving packet $p_j$ is randomly selected from $1$ to $10$. We conduct the simulation by varying the packet reception deadline $T$ in $[5, 35]$.

As shown in Fig.~\ref{sim.theoretic}, the total benefit obtained with our greedy Algorithm~\ref{greedy2} is more than $(1-\frac{1}{e})$ of optimal solution to pairwise coding, which verifies our analysis result. We also can find that the total benefit achieved with our algorithm increases with the increase of the packet reception deadline. This is reasonable since more packets can be timely scheduled before missing the deadline, when the packet reception deadline increases.

\begin{figure}[t]
\begin{center}
\includegraphics[height=48mm,width=65mm]{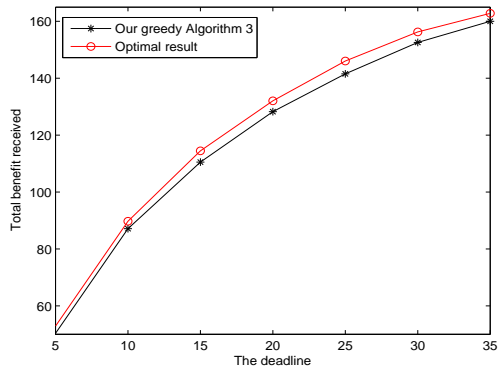}
\caption{The impact of packet reception deadline on the benefits that can be achieved with pairwise coding.}\label{sim.theoretic}
\end{center}
\end{figure}

\section{Conclusion}\label{Sec.conclusion}
In this paper, we propose a novel joint rate selection and network coding (RSNC) scheme for time critical applications. We first prove that the proposed problem is NP-hard, and design a novel graph model to model the problem. Using the graph model, we mathematically formulate the problem. We also propose a metric, based on which we design a heuristic algorithm to determine transmission rate and coding strategy for each transmission so as to maximize the total benefit by timely receiving the packets. The benefit can be defined based on the priority or importance of a packet. We then study the pairwise coding for a special case when all the deadlines for all the packets are the same, and design an efficient algorithm which can achieve at least $1-\frac{1}{e}$ of the optimal solution for pairwise coding. Finally, simulation results demonstrate the proposed RSNC algorithm effectively increases the benefit and ensures timely packet receptions at the destination nodes.


\appendices
\section{Proof of Lemma~\ref{lemma1}}\label{appendix_lemma1}

We prove it by induction.

When $k=1$, obviously the above lemma is true. That is, $Q_{i'_1}=Q_{i^*_1}$, and thus $\mathcal{Q}'_1=\mathcal{Q}^*_1=\{Q_{j'_{1}}\}=\{Q_{j^*_{1}}\}=\{Q_{i^*_1}\}$.
For the case when $k=2$, from the algorithms, we know that $Q_{i'_2}$ and $Q_{i^*_2}$ maximize the ratio $\frac{w'_1(Q_{i'_2})}{c_{i'_2}}$ and $\frac{w^*_1(Q_{i^*_2})}{c_{i^*_2}}$, among the rest cliques in $\mathcal{Q}-\mathcal{Q}'_1$ and $\mathcal{Q}-\mathcal{Q}^*_1$ respectively. Since $\mathcal{Q}'_1=\mathcal{Q}^*_1$, $Q_{i'_2}$ and $Q_{i^*_2}$ are the same.


Without loss of generality, we assume that the above lemma is true for each of the first $k$ iterations, where $1\leq k< \min\{l'_{q'}, l^*_{q^*}\}$. Let $U'$ and $U^*$ be the collections of unconsidered cliques left after the $k$-th iteration with the algorithm in \cite{Khuller1999} and Algorithm~\ref{greedy2} respectively.
We then consider the $(k+1)$-th iteration.

According to the definition, $Q_{i'_{k+1}}$ and $Q_{i^*_{k+1}}$ maximize the ratio $\frac{w'_{k}(Q_{i'_{k+1}})}{c_{i'_{k+1}}}$ and $\frac{w^*_{k}(Q_{i^*_{k+1}})}{c_{i^*_{k+1}}}$, among the left cliques in $U'$ and $U^*$ respectively. Firstly, we can obtain that $U^*\subseteq U'$. This is because, at each of the first $k$ iterations, our Algorithm~\ref{greedy2} may delete two cliques from $U^*$: one is the current considered clique, and the other is the selected clique in the current iteration, while the algorithm of \cite{Khuller1999} only deletes the current considered clique. According to the above assumption, the cliques considered in each of the first $k$ iterations with these two algorithms are the same.
In addition, for each clique $Q_h$ in $U'-U^*$, the weight of the elements covered by $Q_h$ but not covered by any clique in $\mathcal{Q}'_k$ must be 0, which thus does not need to be considered. In other words, the clique $Q_{i'_{k+1}}$ that maximizes the ratio of weight $w'_{k}(Q_{i'_{k+1}})$ to cost among cliques in $U'$ should be the same as the clique $Q_{i^*_{k+1}}$ that maximizes the ratio of weight $w^*_{k}(Q_{i^*_{k+1}})$ to cost among cliques in $U^*$. Thus, the above lemma holds for the $(k+1)$-th iteration.

To summarize, we can obtain that at each $k$-th iteration, where $1\leq k\leq\min\{l'_{q'}, l^*_{q^*}\}\}$, the cliques considered by the algorithm in \cite{Khuller1999} and Algorithm~\ref{greedy2} are the same, i.e., $Q_{i'_{k}}=Q_{i^*_{k}}$.

\section{Proof of Lemma~\ref{lemma2}}\label{appendix_lemma2}

Note that according to Lemma~\ref{lemma1}, at the $k$-th iteration, we have $Q_{i'_k}=Q_{i^*_k}$. We then prove the above lemma by considering the following two cases.

Case 1: If algorithm in \cite{Khuller1999} selects clique $Q_{i'_k}$, there are two cases for Algorithm~\ref{greedy2}. If there is no common element between clique $Q_{i^*_k}$ and the cliques in $\mathcal{Q}^*_{k-1}$, $Q_{i^*_k}$ must be also selected. Otherwise, there must exist another clique, which includes the elements covered by clique $Q_{i^*_k}$ but not covered by any clique in $\mathcal{Q}^*_{k-1}$. This is because if there does not exists such a clique, the weight $w^*_{k}(Q_{i^*_{k+1}})$ must be 0, which contradicts that the clique $Q_{i^*_k}$ maximizes the ratio of the weight $w^*_{k}(Q_{i^*_{k+1}})$ to cost. Thus, Algorithm~\ref{greedy2} must also select a clique.

Case 2: If Algorithm in~\cite{Khuller1999} cannot add in clique $Q_{i'_k}$, Algorithm~\ref{greedy2} will also discard the clique and select nothing. This is because,
the only reason that the Algorithm in ~\cite{Khuller1999} cannot select the current clique is that adding the current clique exceeds the delay constraint $T$, which also needs to be satisfied by our Algorithm~\ref{greedy2}. In this case, both algorithm cannot add in new cliques.

To sum up, we proved Lemma~\ref{lemma2}.
\section{Proof of Lemma~\ref{lemma2.a}}\label{appendix_lemma2.a}

We prove it by induction.

Firstly, when $k=1$, we can easily obtain that the cliques in $\mathcal{Q}'_1$ and $\mathcal{Q}^*_1$ cover the same set of elements, i.e., $S'_1=S^*_1$.
Without loss of generality, we assume that after the first $k-1$ iterations, the cliques in $\mathcal{Q}'_{k-1}$ and $\mathcal{Q}^*_{k-1}$ cover the same set of the elements. That is, $S'_{k-1}=S^*_{k-1}$.  We then consider the case of the $k$-th iteration.

As in Lemma~\ref{lemma2}, at the $k$-th iteration, if algorithm in \cite{Khuller1999} selects the clique $Q_{i'_k}$, there are two cases for Algorithm~\ref{greedy2}: 1) select the clique $Q_{i^*_k}$, or 2) add the clique $Q_{h}$ that includes the elements covered by $Q_{i^*_k}$ but not covered by any clique in $\mathcal{Q}^*_{k-1}$, i.e., $Q_{h}=Q_{i^*_k}-Q_{i^*_k}\bigcap S^*_{k-1}$. For the first case, according to Lemma~\ref{lemma1}, we have $Q_{i'_k}=Q_{i^*_k}$. As $S'_{k-1}=S^*_{k-1}$, we can obtain that
\begin{eqnarray}
S^*_k&=&S^*_{k-1}\bigcup Q_{i^*_k}\notag\\
&=&S'_{k-1}\bigcup Q_{i'_k}\notag\\
&=&S'_k
\end{eqnarray}
For the second case, we can obtain that
\begin{eqnarray}
S^*_k&=&S^*_{k-1}\bigcup Q_h\notag\\
&=&S^*_{k-1}\bigcup (Q_{i^*_k}-Q_{i^*_k}\bigcap S^*_{k-1})\notag\\
&=&S^*_{k-1}\bigcup Q_{i^*_k}\notag\\
&=& S'_{k-1}\bigcup Q_{i'_k}=S'_k
\end{eqnarray}
Thus, in either case, $\mathcal{Q}^*_{k-1}\bigcup \{Q_{i^*_k}\}$ must cover the same set of elements as $\mathcal{Q}^*_{k-1}\bigcup \{Q_{h'}\}$, i.e., $S'_k=S^*_k$.

If at the $k$-th iteration, algorithm in \cite{Khuller1999} does not add in a clique, according to Lemma~\ref{lemma2}, Algorithm~\ref{greedy2} must also not select any new clique. That is, $\mathcal{Q}'_k=\mathcal{Q}'_{k-1}$, $\mathcal{Q}^*_k=\mathcal{Q}^*_{k-1}$. According to the above assumption, we have $S'_k=S^*_k$. .

Thus, after the $k$-th iteration, the cliques in $\mathcal{Q}'_k$ and $\mathcal{Q}^*_k$ cover the same set of the elements, which thus proves Lemma~\ref{lemma2.a}.

\section{Proof of Lemma~\ref{lemma3}}\label{appendix_lemma3}

According to lemma~\ref{lemma2}, at the $k$-th iteration, where $1\leq k\leq\min\{l'_{q'}, l^*_{q^*}\}$, if algorithm in \cite{Khuller1999} adds in clique $Q_{i'_k}$, our Algorithm~\ref{greedy2} must also select clique $Q_{i^*_k}$ or clique $Q_h=Q_{i^*_k}-Q_{i^*_k}\bigcap S^*_{k-1}$, which is a subset of clique $Q_{i^*_k}$. According to Lemma~\ref{lemma1}, we have $Q_{i'_k}=Q_{i^*_k}$.

Firstly, if Algorithm \ref{greedy2} selects $Q_{i^*_k}$, the cost added by both algorithms must be the same. Secondly, if Algorithm~\ref{greedy2} selects another clique $Q_h$, which includes the subset of elements in $Q_{i^*_k}$, the cost added by Algorithm~\ref{greedy2} must be no more than that added by clique $Q_{i'_k}$. In other words, at the $k$-th iteration, the cost added by selecting a clique with algorithm in \cite{Khuller1999} is no less than that by our Algorithm~\ref{greedy2}.

Thus, Algorithm~\ref{greedy2} can add in more cliques before exceeding the constraint $T$, i.e., $q'\leq q^*$. In addition, according to Lemma \ref{lemma2}, we have $l_{q'}\leq l_{q^*}$.

\section{Proof of Lemma~\ref{lemma4}}\label{appendix_lemma4}

From Lemma~\ref{lemma2.a}, we can obtain that after the first $k=\min\{l'_{q'}, l^*_{q^*}\}=l'_{q'}$ iterations, the collections of the selected cliques with both algorithms, e.g., $\mathcal{Q}'_k$ and $\mathcal{Q}^*_k$, cover the same set of elements. In other words, $W(\mathcal{Q}'_k)=W(\mathcal{Q}^*_k)$.

In addition, because $l'_{q'}\leq l^*_{q^*}$, after $k=l'_{q'}$ iterations, the Algorithm~\ref{greedy2} still has a chance to select more $q^*-q'$ cliques. Hence, the total weight of the elements covered by $\mathcal{Q}'$ must be less than the total weight of the elements covered by $\mathcal{Q}^*$, i.e., $W(\mathcal{Q}')\leq W(\mathcal{Q}^*)$.

\section{Proof of Theorem~\ref{theorem.ratio}}\label{appendix_theorem.ratio}

Assume that the optimal solution of budgeted set coverage problem in \cite{Khuller1999} is $OPT'$, while the optimal solution of our pairwise coding is $OPT^*$.
The only difference between these two problems is that in our problem, the selected cliques should be disjoint with each other. Note that the optimal solution to our problem is only one of the feasible solutions to the problem in \cite{Khuller1999}. Thus, $OPT'\geq OPT^*$.

With Lemma~\ref{lemma4}, we have $W(\mathcal{Q}^*)\geq W(\mathcal{Q}')$. According to \cite{Khuller1999}, we can obtain
\begin{eqnarray}
W(\mathcal{Q}^*)&\geq& W(\mathcal{Q}')\notag\\
&\geq& (1-\frac{1}{e})OPT'\notag\\
&\geq& (1-\frac{1}{e})OPT^*
\end{eqnarray}

Thus, the total weight achieved with our algorithm is at least $1-\frac{1}{e}$ of the optimal solution for pairwise coding.

As maximizing the total benefit of the packets that are received without missing the deadlines is equivalent to finding a collection of the cliques with maximum weight before deadline $T$, the total benefit achieved with our algorithm is thus at least $1-\frac{1}{e}$ of the optimal solution for pairwise coding.

\begin{biography}[{\includegraphics[width=1\textwidth]{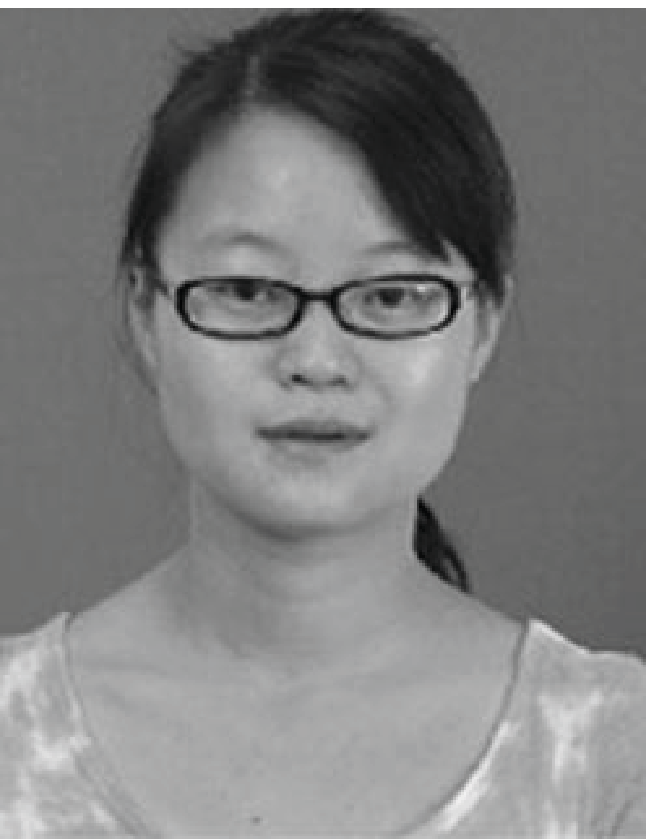}}]{Xiumin Wang}
received her B.S. from the Department of Computer Science, Anhui
Normal University, China, in 2006, and joint Ph.D. degree from
School of Computer Science and Technology of University of Science
and Technology of China and City University of Hong Kong. She did
her postdoc at Singapore University of Technology and Design from
2011 to 2012. Currently, she is with the School of Computer and
Information, Hefei University of Technology. Her research interests
include wireless networks, routing design, and network coding.
\end{biography}

\begin{biography}[{\includegraphics[width=1\textwidth]{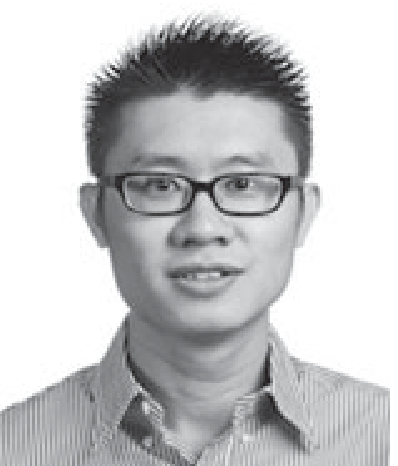}}]{Chau Yuen}
received the B.Eng. and Ph.D. degrees from Nanyang Technological
University, Singapore, in 2000 and 2004, respectively. In 2005, he
was a Postdoctoral Fellow with Lucent Technologies Bell Labs, Murray
Hill, NJ, USA. In 2008, he was a Visiting Assistant Professor with
Hong Kong Polytechnic University, Kowloon, Hong Kong. From 2006 to
2010, he was with the Institute for Infocomm Research, Singapore, as
a Senior Research Engineer. Since 2010, he has been an Assistant
Professor with the Singapore University of Technology and Design,
Singapore. Dr. Yuen serves as an Associate Editor for the IEEE
TRANSACTIONS ON VEHICULAR TECHNOLOGY. He received the IEEE Asia
Pacific Outstanding Young Researcher Award in 2012.
\end{biography}

\begin{biography}[{\includegraphics[width=1\textwidth]{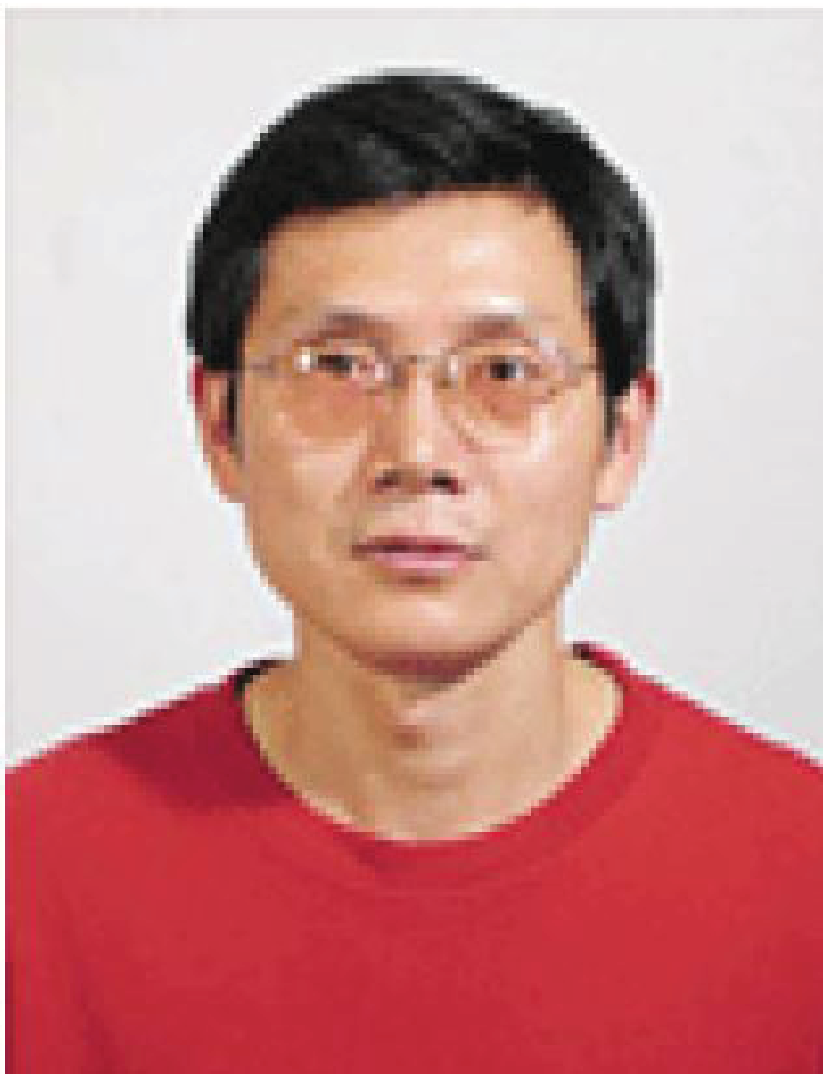}}]{Yinlong Xu}
received his B.S. in Mathematics from Peking University in 1983, and
MS and Ph.D in Computer Science from University of Science and
Technology of China(USTC) in 1989 and 2004 respectively. He is
currently a professor with the School of Computer Science and
Technology at USTC. Prior to that, he served the Department of
Computer Science and Technology at USTC as an assistant professor, a
lecturer, and an associate professor. Currently, he is leading a
group of research students in doing some networking and high
performance computing research. His research interests include
network coding, wireless network, combinatorial optimization, design
and analysis of parallel algorithm, parallel programming tools, etc.
He received the Excellent Ph.D Advisor Award of Chinese Academy of
Sciences in 2006.
\end{biography}

\end{document}